\newcommand{\longversion}[1]{#1}
\newcommand{\shortversion}[1]{}

\shortversion{\documentclass[envcountsame,10pt,letterpaper]{article}}
\longversion{
\documentclass[10pt,a4paper]{article}
\usepackage[hscale=0.7,scale=0.75]{geometry}
}

\shortversion{
\usepackage{aaai}
\usepackage{times}
\usepackage{helvet}
\usepackage{courier}
\setlength{\pdfpagewidth}{8.5in}
\setlength{\pdfpageheight}{11in}
  \frenchspacing
}

\longversion{
\usepackage{url} \urlstyle{rm}

\pagenumbering{arabic}
\pagestyle{plain}
\nonfrenchspacing
}

\newcommand{\citex}[1]{\citeauthor{#1}~\shortcite{#1}}
\newcommand{\citey}[1]{\citeauthor{#1},~\citeyear{#1}}

\shortversion{ 
  \title{Backdoors to Normality for Disjunctive Logic
    Programs\thanks{Research supported by ERC (COMPLEX REASON
      239962).
}
  }
  \author{Johannes Klaus Fichte \and Stefan Szeider\\
    Vienna University of Technology, Austria\\
    {fichte@kr.tuwien.ac.at}, {stefan@szeider.net}}
  \pdfinfo{ 
    /Title (Backdoors to Normality for Disjunctive Logic Programs) 
    /Author (Johannes Klaus Fichte, Stefan Szeider) 
  }
}

\longversion{ \title{Backdoors to Normality for\\ Disjunctive Logic
    Programs\thanks{Research supported by the ERC, Grant COMPLEX
      REASON 239962.}  \footnote{This is the author's self-archived
      copy including detailed proofs. A preliminary version of the
      paper was presented on the workshop ASPOCP'12.}}
  
  \author{Johannes Klaus Fichte and Stefan Szeider\\[3pt]
    Vienna University of Technology, Austria\\
    {fichte@kr.tuwien.ac.at}, {stefan@szeider.net} }
}

\usepackage{amsmath,amssymb,amsthm}
\longversion{\usepackage{amsfonts}}
\usepackage{tikz}

\usepackage{xspace}              
\usepackage{microtype}

\newtheorem{LEM}{Lemma} 
 
\newtheorem{THE}{Theorem} 
\newtheorem{COR}{Corollary}

\newtheorem{PRO}{Proposition} 

\shortversion{\theoremstyle{remark}}
\newtheorem{EX}{Example} 

\def\hy{\hbox{-}\nobreak\hskip0pt} 
\def\hyph{-\penalty0\hskip0pt\relax}

\newcommand{\SB}{\{\,}%
\newcommand{\SM}{\;{|}\;}%
\newcommand{\SE}{\,\}}%

\newcommand{\Card}[1]{|#1|}
\newcommand{\CCard}[1]{\|#1\|}
\let\phi=\varphi
\let\epsilon=\varepsilon

\newcommand{\CCC}{\mathcal{C}}

\newcommand{\NP}{\text{\normalfont NP}}
\newcommand{\coNP}{\text{\normalfont co-NP}}
\newcommand{\paraNP}{\text{\normalfont paraNP}}
\newcommand{\coparaNP}{\text{\normalfont co-paraNP}}
 
\newcommand{\FPT}{\text{\normalfont FPT}}

\newcommand{\Nat}{\mathbb{N}}

\newcommand{\ta}[1]{\text{\normalfont ta(\ensuremath{#1})}}

\let\phi=\varphi
\newcommand{\stableset}{\text{\normalfont AS}}

\newcommand{\at}{\text{\normalfont at}}

\newcommand{\pnot}{\neg}
\newcommand{\por}{\vee}
\newcommand{\rsep}{;\;}
\newcommand{\lrsep}{}

\renewcommand{\DH}{\text{Pos}}
\newcommand{\Constraints}{\text{Constr}}

\newcommand{\MinCheck}{\textsc{MinCheck}\xspace}

\newcommand{\false}{\textit{false}\xspace}

\newcommand{\delBds}[1]{deletion {\ensuremath{#1}}\hy backdoor}
\newcommand{\strongBds}[1]{strong {\ensuremath{#1}}\hy backdoor}

\newcommand{\problemn}[1]{{\scshape #1}}
\newcommand{\BdBrave}[1]{\problemn{\ensuremath{#1}\hy Back\-door\hyph
    Brave\hyph Rea\-son\-ing}}
\newcommand{\BdSkept}[1]{\problemn{\ensuremath{#1}\hy Back\-door\hyph
    Skeptical\hyph Rea\-son\-ing}} 
\newcommand{\BdSkeptHyph}[1]{\problemn{\ensuremath{#1}\hyph Back\-door\hyph
    Skeptical\hyph Rea\-son\-ing}} 
\newcommand{\BdCheck}[1]{\problemn{\ensuremath{#1}\hy Back\-door\hyph Asp\hyph Check}}
\newcommand{\BdDetect}[2]{\problemn{#1 \ensuremath{#2}\hy Back\-door\hyph De\-tec\-tion}}

\newcommand{\SAT}{\problemn{Sat}\xspace} 
\newcommand{\ASP}{\textsc{Asp}\xspace}

\newcommand{\targetclass}[2][]{%
\ifx#10{{\normalfont\textbf{#2}}}\else{{\normalfont\textbf{#2}}}\fi}
\newcommand{\Normal}[1][]{\normalfont{\textbf{Normal}}\xspace}
\newcommand{\tight}[1][]{\targetclass[#1]{Tight}\xspace}

\longversion{
\usepackage{mfirstuc}
\newcommand{\pproblem}[4]{
\begin{quote}
\begin{samepage}
{\scshape {#1}\nopagebreak[4]} \normalfont \vspace{0.4em}\nopagebreak[4]\\
\begin{tabular}{lp{0.6\textwidth}}
  \emph{Given:} & \makefirstuc{#2} \tabularnewline[1pt]
 \emph{Parameter:} & \makefirstuc{#3} \tabularnewline[1pt]
  \emph{Question:} & \makefirstuc{#4} \tabularnewline[1pt]
\end{tabular}
\end{samepage}
\end{quote}
}
}

\usepackage{enumitem} 
\usepackage{booktabs}

\newcommand{\backdoorX}{X}
\newcommand{\modelM}{M}

\newcommand{\FSkept}{F_{\text{Skept}}}
\newcommand{\FBrave}{F_{\text{Brave}}}

\makeatletter
\longversion{
  \def\leftcite{\@up[}\def\rightcite{\@up]}
  
  \def\cite{\def\citeauthoryear##1##2{\def\@thisauthor{##1}%
               \ifx \@lastauthor \@thisauthor \relax \else##1, \fi ##2}\@icite}
  \def\shortcite{\def\citeauthoryear##1##2{##2}\@icite}
  
  \def\citeauthor{\def\citeauthoryear##1##2{##1}\@nbcite}
  \def\citeyear{\def\citeauthoryear##1##2{##2}\@nbcite}
  
  \def\@icite{\leavevmode\def\@citeseppen{-1000}%
   \def\@cite##1##2{\leftcite\nobreak\hskip 0in{##1\if@tempswa , ##2\fi}\rightcite}%
   \@ifnextchar [{\@tempswatrue\@citex}{\@tempswafalse\@citex[]}}
  \def\@nbcite{\leavevmode\def\@citeseppen{1000}%
   \def\@cite##1##2{{##1\if@tempswa , ##2\fi}}%
   \@ifnextchar [{\@tempswatrue\@citex}{\@tempswafalse\@citex[]}}
  
  \def\@citex[#1]#2{%
    \def\@lastauthor{}\def\@citea{}%
    \@cite{\@for\@citeb:=#2\do
      {\@citea\def\@citea{;\penalty\@citeseppen\ }%
       \if@filesw\immediate\write\@auxout{\string\citation{\@citeb}}\fi
       \@ifundefined{b@\@citeb}{\def\@thisauthor{}{\bf ?}\@warning
         {Citation `\@citeb' on page \thepage \space undefined}}%
       {\csname b@\@citeb\endcsname}\let\@lastauthor\@thisauthor}}{#1}}
  
  \def\@biblabel#1{\def\citeauthoryear##1##2{##1, ##2}\@up{[}#1\@up{]}\hfill}
  
  \def\@up#1{\leavevmode\raise.2ex\hbox{#1}}
}
\makeatother

\begin{document}
\maketitle
\begin{abstract}
  Over the last two decades, propositional satisfiability (\SAT) has
  become one of the most successful and widely applied techniques for
  the solution of NP-complete problems. The aim of this paper is to
  investigate theoretically how \SAT can be utilized for the efficient
  solution of problems that are harder than NP or co-NP.  In
  particular, we consider the fundamental reasoning problems in
  propositional disjunctive answer set programming (\ASP),
  \textsc{Brave Reasoning} and \textsc{Skeptical Reasoning}, which ask
  whether a given atom is contained in at least one or in all answer
  sets, respectively.  Both problems are located at the second level
  of the Polynomial Hierarchy and thus assumed to be harder than NP or
  co-NP\shortversion{\@}. One cannot transform these two reasoning
  problems into \SAT in polynomial time, unless the Polynomial
  Hierarchy collapses.

  We show that certain structural aspects of disjunctive logic
  programs can be utilized to break through this complexity barrier,
  using new techniques from Parameterized Complexity. In particular,
  we exhibit transformations from \textsc{Brave} and \textsc{Skeptical
    Reasoning} to \SAT that run in time $O(2^k n^2)$ where~$k$ is a
  structural parameter of the instance and $n$ the input size. In
  other words, the reduction is fixed-parameter tractable for
  parameter $k$. As the parameter $k$ we take the size of a smallest
  backdoor with respect to the class of normal (i.e.,
  disjunction-free) programs. Such a backdoor is a set of atoms that
  when deleted makes the program normal. In consequence, the
  combinatorial explosion, which is expected when transforming a
  problem from the second level of the Polynomial Hierarchy to the
  first level, can now be confined to the parameter~$k$, while the
  running time of the reduction is polynomial in the input size $n$,
  where the order of the polynomial is independent of $k$. We show
  that such a transformation is not possible if we consider backdoors
  with respect to tightness instead of normality.

  We think that our approach is applicable to many other hard
  combinatorial problems that lie beyond NP or co-NP, and thus
  significantly enlarge the applicability of \SAT.

\end{abstract}
\section{Introduction}                                   
Over the last two decades, propositional satisfiability (\SAT) has
become one of the most successful and widely applied techniques for
the solution of $\NP$\hy complete problems. Today's \SAT-solvers are
extremely efficient and robust, instances with hundreds of thousands
of variables and clauses can be solved routinely. In fact, due to the
success of \SAT, $\NP$\hy complete problems have lost their scariness,
as in many cases one can efficiently encode $\NP$\hy complete problems
to \SAT and solve them by means of a
\SAT-solver~\cite{GomesKautzSabharwalSelman08,BiereHeuleMaarenWalsh09}.

We investigate transformations into \SAT for problems that are harder
than $\NP$ or $\coNP$. In particular, we consider various search
problems that arise in \emph{disjunctive answer set programming
  (\ASP)}. With \ASP one can describe a problem by means of rules that
form a disjunctive logic program, whose solutions are answer
sets. Many important problems of AI and reasoning can be represented
in terms of the search for answer
sets~\cite{BrewkaEiterTruszczynski11,MarekTruszczynski99,Niemela99}. Two
of the most fundamental \ASP problems are \textsc{Brave Reasoning} (is
a certain atom contained in at least one answer set?) and
\textsc{Skeptical Reasoning} (is a certain atom contained in all answer
sets?). Both problems are located at the second level of the
Polynomial Hierarchy~\cite{EiterGottlob95} and thus assumed to be
harder than $\NP$ or $\coNP$. It would be desirable to utilize
\SAT-solvers for these problems. However, we cannot transform these
two reasoning problems into \SAT in polynomial time, unless the
Polynomial Hierarchy collapses, which is believed to be unlikely.

\paragraph{New Contribution} In this work we show how to utilize
certain structural aspects of disjunctive logic programs to transform
the two \ASP reasoning problems into \SAT. In particular, we exhibit a
transformation to \SAT that runs in time $O(2^k n^2)$ where~$k$ is a
structural parameter of the instance and $n$ is the input size of the
instance. Thus the combinatorial explosion, which is expected when
transforming problems from the second level of the Polynomial
Hierarchy to the first level, is confined to the parameter~$k$, while
the running time is polynomial in the input size $n$ and the order of
the polynomial is independent of~$k$. Such transformations are known
as ``\emph{fpt-transformations}'' and form the base of the
completeness theory of \emph{Parameterized
  Complexity}~\cite{DowneyFellows99,FlumGrohe06}.  Our reductions
\emph{break complexity barriers} as they move problems form the second
to the first level of the Polynomial Hierarchy.

It is known that the two reasoning problems, when restricted to
so-called \emph{normal programs}, drop to $\NP$ and
$\coNP$~\cite{BidoitFroidevaux91,MarekTruszczynski91a,MarekTruszczynski91},
respectively.  Hence, it is natural to consider a structural
parameter~$k$ as the distance of a given program from being normal. We
measure the distance in terms of the smallest number of atoms that
need to be deleted to make the program
normal. Following~\citex{WilliamsGomesSelman03} we call such a set of
deleted atoms a \emph{backdoor}. We show that in time $O(2^k n^2)$ we
can solve both of the following two tasks for a given program $P$ of
input size $n$ and an atom $a^*$:

\emph{Backdoor Detection:} Find a backdoor of size at most $k$ of the
given program~$P$, or decide that a backdoor of size $k$ does not
exist.

\emph{Backdoor Evaluation:} Transform the program $P$ into two
propositional formulas $\FBrave(a^*)$ and $\FSkept(a^*)$ such that
(i)~$\FBrave(a^*)$ is satisfiable if and only if $a^*$ is in some
answer set of $P$, and (ii)~$\FSkept(a^*)$ is unsatisfiable if and
only if $a^*$ is in all answer sets of $P$.

\emph{Tightness} is a property of disjunctive logic programs that,
similar to normality, lets the complexities of \textsc{Brave} and
\textsc{Skeptical Reasoning} drop to $\NP$ and $\coNP$,
respectively~\cite{Clark78,Fages94}.  Consequently, one could also
consider backdoors to tightness. We show, however, that the reasoning
problems already reach their full complexities (i.e., completeness for
the second level of the Polynomial Hierarchy) with programs of
distance one from being tight. Hence, an fpt-transformation into \SAT
for programs of distance $k>0$ from being tight is not possible unless
the Polynomial Hierarchy collapses.

\paragraph{Related Work} 
Williams, Gomes, and Selman~\shortcite{WilliamsGomesSelman03}
introduced the notion of backdoors to explain favorable running times
and the heavy-tailed behavior of \SAT and CSP solvers on practical
instances. The parameterized complexity of finding small backdoors was
initiated by Nishimura, Ragde, and
Szeider~\shortcite{NishimuraRagdeSzeider04-informal}. For further
results regarding the parameterized complexity of problems related to
backdoors for \SAT, we refer to a recent survey
paper~\cite{GaspersSzeider12a}. Fichte and
Szeider~\shortcite{FichteSzeider12a} formulated a backdoor approach for
\ASP problems, and obtained complexity results with respect to the
target class of Horn programs and various target classes based on
acyclicity; some results could be generalized~\cite{Fichte12}. Both
papers are limited to target classes where we can enumerate the set of
all answer sets in polynomial time. The results do not carry over to
the present work since here we consider target classes where the
problem of determining an answer set is already \NP-hard.

Translations from \ASP problems to \SAT have been explored by several
authors; existing research mainly focuses on transforming programs for
which the reasoning problems already belong to $\NP$ or $\coNP$. In
particular, translations have been considered for head cycle free
programs \cite{Ben-EliyahuDechter94}, tight programs \cite{Fages94},
and normal programs \cite{LinZhao04a,Janhunen06}.

Some authors have generalized the above translations to capture
programs for which the reasoning problems are outside $\NP$ and
$\coNP$.  Janhunen~et~al.~\shortcite{JanhunenEtAl06} considered
programs where the number of disjunctions in the heads of rules is
bounded. They provided a translation that allows a \SAT encoding of
the test whether a candidate set of atoms is indeed an answer set of
the input program. Lee and Lifschitz~\shortcite{LeeLifschitz03}
considered programs with a bounded number of cycles in the positive
dependency graph. They suggested a translation that, similar to ours,
transforms the input program into an exponentially larger
propositional formula whose satisfying assignments correspond to
answer sets of the program. As pointed out by Lifschitz and
Razborov~\shortcite{LifschitzRazborov06}, this translation produces an
exponential blowup already for normal programs (we note that by way of
contrast, our translation is in fact quadratic for normal programs).

Over the last few years, several \SAT techniques have been integrated
into practical \ASP solvers. In particular, solvers for normal
programs (Cmodels~\cite{GiunchigliaLierlerMaratea06},
ASSAT~\cite{LinZhao04a}, Clasp~\cite{GebserKaufmannNeumannSchaub07})
use certain extensions of Clark's completion and then utilize either
black box \SAT solvers or integrate conflict analysis, backjumping,
and other techniques within the \ASP
context. ClaspD~\cite{DrescherEtAl08} is a disjunctive \ASP-solver
that utilizes nogoods based on the logical characterizations of loop
formulas~\cite{Lee05}.

\section{Preliminaries}

\paragraph{Answer set programs} We consider a universe of
propositional \emph{atoms}. A \emph{disjunctive logic program} (or
simply a \emph{program})~$P$ is a set of \emph{rules} of the form
$x_1\por \dots \por x_l \leftarrow y_1,\dots,y_n,\pnot z_1, \dots,
\pnot z_m \lrsep$ where $x_1,\dots,x_l,$ $y_1,\dots,y_n,$ $z_1,\dots,
z_m$ are atoms and $l,n,m$ are non-negative integers. We write
$H(r)=\{x_1, \dots, x_l\}$ (the \emph{head} of $r$), $B^+(r)= \{y_1,
\dots, y_n\}$ (the \emph{positive body} of $r$), and $B^-(r) = \{z_1,
\dots, z_m\}$ (the \emph{negative body} of $r$). We denote the sets of
atoms occurring in a rule~$r$ or in a program~$P$ by $\at(r)=H(r) \cup
B^+(r) \cup B^-(r)$ and $\at(P)=\bigcup_{r\in P} \at(r)$,
respectively.\longversion{ We abbreviate the number of rules of $P$ by
  $\Card{P} = \Card{\SB r \SM r \in P\SE}$.} A rule~$r$ is
\emph{negation-free} if $B^-(r)=\emptyset$, $r$ is \emph{normal} if
$\Card{H(r)}\leq 1$, $r$ is a \emph{constraint} if $\Card{H(r)}=0$,
$r$ is \emph{constraint-free} if $\Card{H(r)}>0$, $r$ is \emph{Horn}
if it is negation-free and normal, $r$ is \emph{positive} if it is
Horn and constraint-free, and $r$ is \emph{tautological} if $B^+(r)
\cap (H(r) \cup B^-(r))\neq \emptyset$. We say that a program has a
certain property if all its rules have the property. We denote the
class of all normal programs by \Normal and the class of all Horn
programs by \textbf{Horn}. In the following, we restrict ourselves to
programs that do \emph{not} contain any tautological rules. This
restriction is not significant as tautological rules can be omitted
from a program without changing its answer sets~\cite{BrassDix98}.
\longversion{Note that we state explicitly the differences regarding
  tautologies in the proofs.}

A set $M$ of atoms \emph{satisfies} a rule $r$ if \((H(r)\cup B^-(r))
\cap M \neq \emptyset\) or $B^+(r) \setminus M \neq \emptyset$. $M$ is
a \emph{model} of $P$ if it satisfies all rules of $P$. The \emph{GL
  reduct} of a program~$P$ under a set~$M$ of atoms is the program
$P^M$ obtained from $P$ by first, removing all rules~$r$ with
$B^-(r)\cap M\neq \emptyset$ and second, removing all~$\neg z$ where
$z \in B^-(r)$ from all remaining
rules~$r$~\cite{GelfondLifschitz91}. $M$ is an \emph{answer set} (or
\emph{stable set}) of a program~$P$ if $M$ is a minimal model of
$P^M$. The Emden-Kowalski operator of a program~$P$ and a subset~$A$
of atoms of $P$ is the set~$T_P(A):=\SB a \SM a \in H(r), B^+(r)
\subseteq A, r \in P\SE$. The \emph{least model} $LM(P)$ is the least
fixed point of $T_P(A)$~\cite{Van-EmdenKowalski76}. Note that every
positive program~$P$ has a unique minimal model which equals the least
model~$LM(P)$~\cite{GelfondLifschitz88}.

\shortversion{
\begin{EX}\label{ex:running}
  Consider the program $P = \SB a \por c \leftarrow b\rsep b
  \leftarrow c, \pnot g\rsep c \leftarrow a\rsep b \por c \leftarrow e
  \rsep h \por i \leftarrow g, \pnot c \rsep a \por b\rsep g
  \leftarrow \pnot i \rsep c \lrsep\SE$.  The set $A =\{ b, c, g \}$
  is an answer set of $P$ since $P^A= \SB a \por c \leftarrow b \rsep
  c \leftarrow a\rsep b \por c \leftarrow e \rsep a \por b \rsep g
  \rsep c \lrsep\SE$ and the minimal models of $P^A$ are $\{b,c, g\}$
  and $\{a,c,g\}$.
\end{EX}
}
\longversion{
\begin{EX}\label{ex:running}
  Consider the program 
  \begin{align*}
    P = \{ a \por c & \leftarrow b\rsep & b & \leftarrow c, \pnot g
    \rsep & c & \leftarrow a\rsep\\
    b \por c & \leftarrow e \rsep & h \por i & \leftarrow g, \pnot c
    \rsep & a \por b & \rsep \\
    g & \leftarrow \pnot i \rsep &  c & \lrsep & &\qquad\}.
  \end{align*}
  The set $A =\{ b, c, g \}$ is an answer set of $P$ since $P^A= \SB a
  \por c \leftarrow b \rsep c \leftarrow a\rsep b \por c \leftarrow e
  \rsep a \por b \rsep g \rsep c \lrsep\SE$ and the minimal models of
  $P^A$ are $\{b,c, g\}$ and $\{a,c,g\}$.
\end{EX}
}

The main reasoning problems for \ASP are \textsc{Brave Reasoning}
(given a program $P$ and an atom $a\in \at(P)$, is $a$ contained in
some answer set of $P$?)  and \textsc{Skeptical Reasoning} (given a
program $P$ and an atom $a\in \at(P)$, is $a$ contained in all answer
sets of $P$?).  \textsc{Brave Reasoning} is $\Sigma^P_2$-complete,
\textsc{Skeptical Reasoning} is
$\Pi^P_2$-complete~\cite{EiterGottlob95}.

\paragraph{Parameterized Complexity} We give some basic background on
parameterized complexity. For more detailed information we refer to
other
sources~\longversion{\cite{DowneyFellows99,FlumGrohe06,GottlobSzeider08,Niedermeier06}}\shortversion{\cite{DowneyFellows99,FlumGrohe06}}.
A \emph{parameterized problem} $L$ is a subset of $\Sigma^* \times
\Nat$ for some finite alphabet $\Sigma$. For an instance~$(I,k) \in
\Sigma^* \times \Nat$ we call $I$ the \emph{main part} and $k$ the
\emph{parameter}. $L$ is \emph{fixed-parameter tractable} if there
exists a computable function $f$ and a constant~$c$ such that there
exists an algorithm that decides whether $(I,k)\in L$ in
time~$O(f(k)\CCard{I}^c)$ where $\CCard{I}$ denotes the size of~$I$.
Such an algorithm is called an \emph{fpt-algorithm}.  $\FPT$ is the
class of all fixed-parameter tractable decision problems.

Let $L \subseteq \Sigma^* \times \Nat$ and $L' \subseteq
\Sigma'^*\times \Nat$ be two parameterized problems for some finite
alphabets $\Sigma$ and $\Sigma'$. An \emph{fpt-reduction} $r$ from $L$
to $L'$ is a many-to-one reduction from $\Sigma^*\times \Nat$ to
$\Sigma'^*\times \Nat$ such that for all $I \in \Sigma^*$ we have
$(I,k) \in L$ if and only if $r(I,k)=(I',k')\in L'$ such that $k' \leq
g(k)$ for a fixed computable function $g: \Nat \rightarrow \Nat$ and
there is a computable function $f$ and a constant $c$ such that $r$ is
computable in time $O(f(k)\CCard{I}^c)$ where $\CCard{I}$ denotes the
size of~$I$~\cite{FlumGrohe06}. Thus, an fpt-reduction is, in
particular, an fpt-algorithm. It is easy to see that the class~$\FPT$
is closed under fpt-reductions. We would like to note that the theory
of fixed-parameter intractability is based on
fpt-reductions~\longversion{\cite{DowneyFellows99,FlumGrohe06}}\shortversion{\cite{FlumGrohe06}}.

\paragraph{Propositional satisfiability} A \emph{truth assignment} is
a mapping $\tau:\backdoorX\rightarrow \{0,1\}$ defined for a set
$\backdoorX$ of atoms. For $x\in \backdoorX$ we put $\tau(\pnot x)=1 -
\tau(x)$. By $\ta{\backdoorX}$ we denote the set of all truth
assignments $\tau:\backdoorX\rightarrow \{0,1\}$.  We usually say
\emph{variable} instead of atom in the context of formulas. Given a
propositional formula $F$, the problem \SAT asks whether $F$ is
satisfiable. We can consider \SAT as a parameterized problem by simply
associating with every formula the parameter~$0$.

\section{Backdoors of Programs} 

In the following we give the main notions concerning backdoors for
answer set programming, as introduced by Fichte and
Szeider~\shortcite{FichteSzeider12a}. Let $P$ be a program,
$\backdoorX$ a set of atoms, and $\tau\in \ta{\backdoorX}$. The
\emph{truth assignment reduct} of $P$ under $\tau$ is the logic
program $P_\tau$ obtained from $P$ by removing all rules $r$ for which
at least one of the following holds: (i)~$H(r)\cap \tau^{-1}(1)\neq
\emptyset$, (ii)~$H(r)\subseteq \backdoorX$, (iii)~$B^+(r) \cap
\tau^{-1}(0)\neq \emptyset$, and (iv)~$B^-(r) \cap \tau^{-1}(1)\neq
\emptyset$, and then removing from the heads and bodies of the
remaining rules all literals $v,\pnot v$ with $v\in \backdoorX$. In
the following, let $\CCC$ be a class of programs. We call $\CCC$ to be
\emph{rule induced} if for each $P\in \CCC$, $P'\subseteq P$ implies
$P'\in \CCC$. A set $\backdoorX$ of atoms is a \emph{\strongBds{\CCC}}
of a program $P$ if $P_{\tau}\in \CCC$ for all truth assignments $\tau
\in \ta{\backdoorX}$. 
Given a \strongBds{\CCC}~$X$ of a program $P$, the answer sets of $P$
are among the answer sets we obtain from the truth assignment reducts
$P_\tau$ where $\tau \in X$, more formally $\stableset(P) \subseteq
\SB M\cup \tau^{-1}(1) \SM \tau\in \ta{X\cap\, \at(P)}, M \in
\stableset(P_\tau)\SE$ where $\stableset(P)$ denotes the set of all
answer sets of~$P$.
For a program~$P$ and a set $\backdoorX$ of
atoms we define $P-\backdoorX$ as the program obtained from $P$ by
deleting all atoms contained in $\backdoorX$ and their negations from
the heads and bodies of all the rules of~$P$. A set $\backdoorX$ of
atoms is a \emph{\delBds{\CCC}} of a program $P$ if $P-\backdoorX \in
\CCC$.

\begin{EX}\label{ex:backdoor}
Consider the program $P$ from Example~\ref{ex:running}. The set
$\backdoorX=\{ b,c,h \}$ is a \strongBds{\Normal} since the truth assignment
reducts
\(
	P_{b=0,c=0,h=0}=P_{000}=\SB  i \leftarrow g \rsep a \rsep g \lrsep \leftarrow \pnot i \SE
\),
\(
    P_{001}=P_{010}=P_{011}=P_{101} = \SB a \rsep g \leftarrow \pnot i \lrsep \SE
\),
\(
    P_{100}=\SB a \rsep i \leftarrow g  \rsep g \leftarrow \pnot i \lrsep \SE
\), and
\(
    P_{110}=P_{111}=\SB g \leftarrow \pnot i \lrsep \SE
\)	 
are in the class \Normal.
\end{EX}

In the following we refer to $\CCC$ as the \emph{target class} of the
backdoor. For most target classes $\CCC$, \delBds{\CCC}s are
\strongBds{\CCC}s. For $\CCC=\Normal$ even the opposite direction is true.

\begin{PRO}[\citey{FichteSzeider12a}]\label{lem:rule-induced} 
  If $\CCC$ is rule induced, then every \delBds{\CCC} is a
  \strongBds{\CCC}.
\end{PRO}

\begin{LEM}\label{lem:normal-strong-deletion}
Let $P$ be a program. A set $X$ is a \strongBds{\Normal} of a program~$P$
if and only if it is a \delBds{\Normal} of~$P$.
\end{LEM} 
\shortversion{
\begin{proof}
  We observe that the class of all normal programs is
  rule-induced. Thus the if direction holds by
  Proposition~\ref{lem:rule-induced}. We proceed to show the only-if
  direction. Assume $X$ is a strong $\Normal$-backdoor of $P$.
  Consider a rule $r'\in P-\backdoorX$ which is not tautological. Let
  $r\in P$ be a rule from which $r'$ was obtained in forming
  $P-\backdoorX$. We define $\tau\in \ta{X}$ by setting all atoms in
  $H(r) \cup B^-(r)$ to 0, all atoms in $B^+(r)$ to 1, and all
  remaining atoms in $\backdoorX\setminus \at(r)$ arbitrarily to 0 or
  1. Since $r$ is not tautological, this definition of $\tau$ is
  sound. It remains to observe that $r'\in P_\tau$. Since $\backdoorX$
  is a \strongBds{\Normal} of $P$, the rule $r'$ is normal. Hence the
  lemma follows.
\end{proof}
} 

\longversion{
\begin{proof}
We observe that the class of all normal programs is rule-induced. Thus
the if direction holds by Proposition~\ref{lem:rule-induced}. We
proceed to show the only-if direction. Assume $X$ is a strong
$\Normal$-backdoor of $P$.  Consider a rule $r'\in P-\backdoorX$ which
is not tautological. Let $r\in P$ be a rule from which $r'$ was
obtained in forming $P-\backdoorX$. We define $\tau\in \ta{X}$ by
setting all atoms in $H(r) \cup B^-(r)$ to 0, all atoms in $B^+(r)$ to
1, and all remaining atoms in $\backdoorX\setminus \at(r)$ arbitrarily
to 0 or 1. Since $r$ is not tautological, this definition of $\tau$ is
sound. It remains to observe that $r'\in P_\tau$. Since $\backdoorX$
is a \strongBds{\Normal} of $P$, the rule $r'$ is normal. Hence, the
lemma follows.
\end{proof}
}

Each target class $\CCC$ gives rise to the following problems:
\shortversion{
\begin{itemize}[leftmargin=*]
\item \BdCheck{\CCC}: \emph{Given} a program~$P$, a \strongBds{\CCC}
  $\backdoorX$ of $P$, a set~$\modelM\subseteq \at(P)$, and the
  parameter size of the backdoor~$k=\Card{\backdoorX}$.  \emph{Is}
  $\modelM$ an answer set of $P$?
\item \BdBrave{\CCC}: \emph{Given} a program~$P$, a \strongBds{\CCC}
  $\backdoorX$ of $P$, an atom~$a^*\in \at(P)$, and the parameter size
  of the backdoor~$k=\Card{\backdoorX}$. Does $a^*$ belong to
  \emph{some} answer set \mbox{of $P$?}

\item \BdSkept{\CCC}: \emph{Given}
  a program~$P$, a \strongBds{\CCC} $\backdoorX$ of $P$, an
  atom~$a^*\in \at(P)$, and the parameter size of the backdoor
  $k=\Card{\backdoorX}$. Does $a^*$ belong to \emph{all} answer sets
  of $P$?
\end{itemize}} 
\longversion{ \pproblem{\BdCheck{\CCC}}{A program~$P$, a
    \strongBds{\CCC} $\backdoorX$ of $P$, a set~$\modelM\subseteq
    \at(P)$, and the size of the backdoor~$k=\Card{\backdoorX}$.}{The
    integer~$k$.}{\emph{Is} $\modelM$ an answer set of $P$?}

  \pproblem{\BdBrave{\CCC}}{A program~$P$, a \strongBds{\CCC}
    $\backdoorX$ of $P$, an atom~$a^*\in \at(P)$, and the size of the
    backdoor~$k=\Card{\backdoorX}$.}{The integer~$k$.}{Does $a^*$
    belong to \emph{some} answer set \mbox{of $P$?}  }

  \pproblem{\BdSkept{\CCC}}{A program~$P$, a \strongBds{\CCC}
    $\backdoorX$ of $P$, an atom~$a^*\in \at(P)$, and the size of the
    backdoor $k=\Card{\backdoorX}$.}{The integer $k$.}{Does $a^*$
    belong to \emph{all} answer sets of $P$?}  }
Problems for {\it\delBds{\CCC}s} can be defined similarly.

\section{Using Backdoors}

In this section, we show results regarding the use of backdoors with
respect to the target class $\Normal$. 

\begin{THE}\label{the:fpt-check}
The problem \BdCheck{\ensuremath{\Normal}} is fixed-parameter tractable.
More specifically, given a program $P$ of input size $n$, a
\strongBds{ \ensuremath{\Normal}} $N$ of $P$ of size $k$, and a set
$\modelM\subseteq \at(P)$ of atoms, we can check in time $O(2^kn)$
whether $\modelM$ is an answer set of $P$.
\end{THE}

The most important part for establishing Theorem~\ref{the:fpt-check} is to
check whether a model is a minimal model. In general, this is a $\coNP$\hy
complete task, but in the context of Theorem~\ref{the:fpt-check} we can
achieve fixed-parameter tractability based on the following construction and
lemma.

Let $P$ be a given program, $\backdoorX$ a \strongBds{ \ensuremath{
    \Normal} } of $P$ of size $k$, and let $\modelM\subseteq
\at(P)$. For a set $\backdoorX_1\subseteq \modelM \cap \backdoorX$ we
construct a program $P_{\backdoorX_1 \subseteq \backdoorX}$ as
follows:
%
(i)~remove all rules~$r$ for which $H(r)\cap \backdoorX_1\neq
\emptyset$ and
(ii)~replace for all remaining rules~$r$ the head~$H(r)$ with
$H(r)\setminus \backdoorX$ and the positive body~$B^+(r)$ with
$B^+(r)\setminus \backdoorX_1$.

Recall that by definition we exclude programs with tautological
rules. Since $\backdoorX$ is a \strongBds{\ensuremath{\Normal}} of
$P$, it is also a \delBds{\ensuremath{\Normal}} of $P$ by
Lemma~\ref{lem:normal-strong-deletion}. Hence $P-\backdoorX$ is
normal. Let $r$ be an arbitrarily chosen rule in $P$. Then there is a
corresponding rule $r' \in P-\backdoorX$ and a corresponding rule~$r''
\in P_{\backdoorX_1 \subseteq \backdoorX}$. Since we remove in both
constructions exactly the same literals from the head of every rule,
$H(r')=H(r'')$ holds. Consequently, $P_{\backdoorX_1 \subseteq
  \backdoorX}$ is normal and $P^\modelM_{\backdoorX_1 \subseteq
  \backdoorX}$ is Horn (here $P_{\backdoorX_1\subseteq
  \backdoorX}^\modelM$ denotes the GL-reduct of
$P_{\backdoorX_1\subseteq \backdoorX}$ under~$\modelM$).

For any program $P'$ let $\Constraints(P')$ denote the set of constrains of
$P'$ and $\DH(P') = P' \setminus \Constraints(P')$. If $P'$ is Horn,
$\DH(P')$ has a least model $L$ and $P'$ has a model if and only if $L$ is a
model of $\Constraints(P')$~\cite{DowlingGallier84}.

\newcommand{\AlgTrue}{\textit{True}\xspace}
\newcommand{\AlgFalse}{\textit{False}\xspace}

Let $X$ be a \strongBds{\Normal} of $P$ and $X_1 \subseteq X$. Given
$\modelM\subseteq \at(P)$, the algorithm $\MinCheck(X_1)$ below
performs the following steps:

\shortversion{\begin{quote}}
\begin{enumerate}[itemsep=0pt,leftmargin=*]
\item Return \AlgTrue if $X_1$ is not a subset of $M$.
\item Compute the Horn program $P^\modelM_{\backdoorX_1\subseteq
\backdoorX}$.

\item Compute the least model $L$ of $\DH(
P^\modelM_{\backdoorX_1 \subseteq \backdoorX})$.

\item\label{alg:step-check} 
    Return \AlgTrue if at least one of the following conditions
    holds:

  \begin{enumerate}

  \item[(a)]\label{cond:notsubmod} $L$ is not a model of
    $\Constraints( P^\modelM_{\backdoorX_1 \subseteq \backdoorX})$.
  \item[(b)]\label{cond:minus} $L$ is not a subset of $\backdoorX$,
  \item[(c)] $L \cup \backdoorX_1$ is not a proper subset of $\modelM$,
  \item[(d)]\label{cond:notmodel} $L \cup \backdoorX_1$ is not a model
    of $P^\modelM$.
  \end{enumerate}
  \item Otherwise return \AlgFalse.
\end{enumerate}
\shortversion{\end{quote}}

\begin{LEM}\label{lem:soundness}  
  Let $X$ be a \strongBds{\Normal}. A model $\modelM\subseteq \at(P)$
  of $P^\modelM$ is a minimal model of $P^\modelM$ if and only if
  $\MinCheck(X_1)$ returns \AlgTrue for each set $X_1\subseteq X$.
\end{LEM}
\shortversion{
  \noindent Because of space constraints we omit the lengthy proof.  }
\longversion{
  \begin{proof}
    ($\Rightarrow$). Assume that $\modelM$ is a minimal model of
    $P^\modelM$, and suppose to the contrary that there is some
    $\backdoorX_1\subseteq \modelM \cap \backdoorX$ for which the
    algorithm returns \AlgFalse.  Consequently, none of the conditions
    in Step~4 of the algorithms holds.  That means, the least model
    $L$ of $P^\modelM_{ \backdoorX_1 \subseteq \backdoorX}$ satisfies
    $\Constraints(P^\modelM_{\backdoorX_1\subseteq \backdoorX})$ and
    is therefore a model of $P^\modelM_{ \backdoorX_1 \subseteq
      \backdoorX}$. Moreover, since $L \cup \backdoorX_1 \nsubseteq
    \modelM$ and $L \cup \backdoorX_1$ is a model of $P^\modelM$,
    $\modelM$ cannot be a minimal model of $P^\modelM$, a
    contradiction to our assumption. So we conclude that the algorithm
    succeeds and the only-if direction of the lemma is shown.

    ($\Leftarrow$). Assume that the algorithm returns \AlgTrue for
    each $\backdoorX_1 \subseteq \modelM \cap \backdoorX$. We show
    that $\modelM$ is a minimal model of $P^\modelM$. Suppose to the
    contrary that $P^\modelM$ has a model $\modelM'\subsetneq
    \modelM$.

    We run the algorithm for $\backdoorX_1:=\modelM' \cap
    \backdoorX$. Let $L$ be the least model of
    $\DH(P^\modelM_{\backdoorX_1\subseteq \backdoorX})$.  By
    assumption, the algorithm returns \AlgTrue, hence some of the
    conditions of Step~4 of the algorithm must hold for~$L$. We will
    show, however, that none of the conditions can hold, which will
    yield to a contradiction, and so establish the if direction of the
    lemma, and thus completes its proof.

    First we show that $\modelM' \setminus \backdoorX$ is a model of
    $P^\modelM_{\backdoorX_1\subseteq \backdoorX}$.  Consider a rule
    $r'\in P^\modelM_{\backdoorX_1\subseteq \backdoorX}$ and let $r\in
    P^\modelM$ such that $r'$ is obtained form $r$ by removing
    $\backdoorX$ from $H(r)$ and by removing $\backdoorX_1$ from
    $B^+(r)$.  Since $\modelM'$ is a model of $P^\modelM$, we have
    (i)~$B^+(r) \setminus \modelM' \neq\emptyset$ or (ii)~$H(r)\cap
    \modelM'\neq \emptyset$. Moreover, since $B^+(r')=B^+(r)\setminus
    \backdoorX_1$ and $\backdoorX_1=\modelM'\cap \backdoorX$,
    (i)~implies $\emptyset \neq B^+(r) \setminus \modelM' = B^+(r)
    \setminus \backdoorX_1 \setminus \modelM' = B^+(r')\setminus
    \modelM' \subseteq B^+(r') \setminus (\modelM' \setminus
    \backdoorX)$, and since $H(r)\cap\backdoorX_1=\emptyset$,
    (ii)~implies $\emptyset \neq H(r)\cap\modelM'= H(r)\cap
    (\modelM'\setminus \backdoorX_1)= H(r)\cap (\modelM'\setminus
    \backdoorX)= (H(r)\setminus \backdoorX) \cap (\modelM'\setminus
    \backdoorX)= H(r') \cap (\modelM'\setminus \backdoorX)$.  Hence
    $\modelM' \setminus \backdoorX$ satisfies $r'$. Since $r'\in
    P^\modelM_{\backdoorX_1\subseteq \backdoorX}$ was chosen
    arbitrarily, we conclude that $\modelM' \setminus \backdoorX$ is a
    model of $P^\modelM_{\backdoorX_1\subseteq \backdoorX}$.

    Since $P^\modelM_{\backdoorX_1\subseteq \backdoorX}$ has some
    model (namely $\modelM' \setminus \backdoorX$), the least model
    $L$ of $\DH(P^\modelM_{\backdoorX_1\subseteq \backdoorX})$ must be
    a model of $P^\modelM_{\backdoorX_1\subseteq \backdoorX}$, thus
    Condition~(a) cannot hold for $L$.

    Next we show that the other conditions cannot hold either.
    Since $\modelM'\setminus \backdoorX$ is a model of
    $P^\modelM_{\backdoorX_1 \subseteq \backdoorX}$, as shown above,
    we have $L\subseteq \modelM'\setminus \backdoorX$.
    %
    We obtain $L \subseteq \modelM \setminus \backdoorX$ since
    $\modelM' \setminus \backdoorX \subseteq \modelM \setminus
    \backdoorX$.
    %
    Further, we obtain $L \cup \backdoorX_1 \subsetneq \modelM$ since
    $L \cup \backdoorX_1 \subseteq (\modelM' \setminus \backdoorX)
    \cup \backdoorX_1 = (\modelM' \setminus \backdoorX) \cup
    (\modelM'\cap \backdoorX) =\modelM' \subsetneq \modelM$.
    Hence we have excluded Conditions~(b) and (c), and it remains to
    exclude Condition~(d).

    %
    Consider a rule $r\in P^M$. If $\backdoorX_1 \cap H(r)\neq
    \emptyset$, then $L\cup \backdoorX_1$ satisfies $r$; thus it
    remains to consider the case $\backdoorX_1 \cap H(r) = \emptyset$.
    In this case there is a rule $r'\in P^\modelM_{\backdoorX_1
      \subseteq \backdoorX}$ with $H(r')=H(r)\setminus \backdoorX$ and
    $B^+(r')=B^+(r)\setminus \backdoorX_1$. Since $L$ is a model of
    $P^\modelM_{\backdoorX_1 \subseteq \backdoorX}$, $L$ satisfies
    $r'$. Hence (i)~$B^+(r')\setminus L \neq\emptyset $ or
    (ii)~$H(r')\cap L \neq \emptyset$.  Since $B^+(r')=B^+(r)\setminus
    \backdoorX_1$, (i)~implies that $B^+(r)\setminus ( L\cup
    \backdoorX_1)\neq \emptyset$; and since $H(r')\subseteq H(r)$,
    (ii)~implies that $H(r)\cap (L \cup \backdoorX_1) \neq
    \emptyset$. Thus $L \cup \backdoorX_1$ satisfies~$r$. Since $r\in
    P^M$ was chosen arbitrarily, we conclude that $L \cup
    \backdoorX_1$ is a model of~$P^M$, which excludes also the last
    Condition~(d).
\end{proof}
}

We are now in a position to establish Theorem~\ref{the:fpt-check}.

\shortversion{
\begin{proof}[Proof of Theorem~\ref{the:fpt-check}]
  First, we check whether $\modelM$ is a model of $P^\modelM$. If it
  is, we run the algorithm $\MinCheck(X_1)$ for each set
  $\backdoorX_1\subseteq \backdoorX$. By Lemma~\ref{lem:soundness} the
  algorithm decides whether $\modelM$ is an answer set of $P$. The
  check of whether $\modelM$ is a model of $P^\modelM$ can clearly be
  carried out in linear time. The algorithm $\MinCheck(X_1)$ runs in
  linear time since we can compute the least model of a Horn program
  in linear time~\cite{DowlingGallier84}. As there are at most $2^k$
  sets $\backdoorX_1$ to consider, the total running time is $O(2^k
  n)$ where $n$ denotes the input size of $P$ and $k=\Card{X}$.  We
  obtain fixed-parameter tractability for the parameter~$k$.
\end{proof}
}

\newcommand{\modelN}{N}
\longversion{
\begin{proof}[Proof of Theorem~\ref{the:fpt-check}]
  First we check whether $\modelM$ is a model of $P^\modelM$. If $\modelM$
  is not a model of $P^\modelM$ then it is not an answer set of $P$, and we
  can neglect it. Hence assume that $\modelM$ is a model of $P^\modelM$. Now
  we run the algorithm \MinCheck. By Lemma~\ref{lem:soundness} the algorithm
  decides whether $\modelM$ is an answer set of $P$.

  In order to complete the proof, it remains to bound the running time. The
  check whether $\modelM$ is a model of $P^\modelM$ can clearly be carried
  out in linear time. For each set $\backdoorX_1\subseteq
  \modelM\cap\backdoorX$ the algorithm \MinCheck runs in linear time. This
  follows directly from the fact that we can compute the least
  model of a Horn program in linear time~\cite{DowlingGallier84}. As there
  are at most $2^k$ sets $\backdoorX_1$ to consider, the total running time
  is $O(2^k n)$ where $n$ denotes the input size of $P$ and $k=\Card{X}$.
  Thus, in particular, the decision is fixed-parameter tractable for
  parameter $k$.
\end{proof}
}

\longversion{
\begin{EX}
  Consider the program \(P\) from Example~\ref{ex:running} and the
  backdoor~$X=\{b,c,h\}$ from Example~\ref{ex:backdoor}.  Let
  \(\modelN = \{ a,b,c,g \} \subseteq \at(P)\). Obviously \(\modelN\)
  is a model of \(P\). We apply the algorithm \MinCheck for each
  \(\backdoorX_1\) of \(\{b,c\}\).
  For \(\backdoorX_1=\emptyset\) we obtain \(P^\modelN_{\backdoorX_1
    \subseteq \backdoorX} = \SB a \leftarrow b\rsep \leftarrow a\rsep
  \leftarrow e \rsep a \rsep g \leftarrow \pnot i \lrsep\SE\) and the
  least model \(L=\{a,g\}\) of \(\DH(P^\modelN_{\emptyset \subseteq
    \backdoorX})\). Since Condition~4a holds (\(L\) is not a model of
  \(\Constraints(P^\modelN_{\backdoorX_1 \subseteq \backdoorX})\)),
  the algorithm returns \AlgTrue.
  For \(\backdoorX_2=\{b\}\) we have \(P^\modelN_{\backdoorX_2
    \subseteq \backdoorX} = \SB a \rsep \leftarrow a \rsep g
  \lrsep\SE\) and \(L=\{g\}\) is the least model of \(\DH( P^\modelN_{
    \backdoorX_2 \subseteq \backdoorX})\). Since Condition~4a holds
  (\(L\) is not a model of \(\Constraints(P^\modelN_{\backdoorX_2
    \subseteq \backdoorX}))\), the algorithm returns \AlgTrue for
  \(X_2\).
  For \(\backdoorX_3=\{c\}\) we obtain \(P^\modelN_{\backdoorX_3
    \subseteq \backdoorX} = \SB a \rsep g \lrsep \SE\). The set
  \(L=\{a , g \}\) is the least model of \(\DH( P^\modelN_{
    \backdoorX_3 \subseteq \backdoorX})\).  Since none of the
  Conditions~4a--d hold, more precisely \(L\) is a model of
  \(\Constraints(P^\modelN_{ \backdoorX_1 \subseteq \backdoorX})\),
  \(L\) is a subset of \(\backdoorX\), \(L \cup \backdoorX_1\) is a
  proper subset of \(\modelN\), and \(L \cup \backdoorX_1\) is a model
  of \(P^\modelN\). Hence, the algorithm returns \AlgFalse. Thus
  \MinCheck does not succeed, and \(\modelM\) is not a minimal model
  of \(P^\modelM\).
\end{EX}
}

\begin{EX}
  \shortversion{Consider}\longversion{Again, consider} the program
  \(P\) from Example~\ref{ex:running} and the backdoor~$X=\{b,c,h\}$
  from Example~\ref{ex:backdoor}. Let \(\modelM = \{ b,c,g \}
  \subseteq \at(P)\). Since \(\modelM\) satisfies all rules in \(P\),
  the set \(\modelM\) is a model of \(P\). We apply the algorithm
  \MinCheck for each subset of \(\{b,c,h\}\).
  For \(\backdoorX_1=\emptyset\) we obtain \(P^{\modelM}_{\backdoorX_1
    \subseteq \backdoorX} = \SB a \leftarrow b \rsep \leftarrow a
  \rsep \leftarrow e \rsep a \rsep g \SE\). The set \(L=\{a, g\}\) is
  the least model of \(\DH(P^{\modelM}_{\backdoorX_1 \subseteq
    \backdoorX})\). Since Condition~4a holds, the algorithm returns
  \AlgTrue for \(X_1\).
  For \(\backdoorX_2=\{b\}\) we have \(P^\modelM_{\backdoorX_2
    \subseteq \backdoorX}=\SB a \rsep \leftarrow a \rsep g \rsep
  \leftarrow \lrsep \SE\) and the least model \(L=\{a, g\}\) of
  \(\DH(P^{\modelM}_{\backdoorX_2 \subseteq \backdoorX})\). Since
  Condition~4a holds, \MinCheck returns \AlgTrue for \(X_2\).
  For \(\backdoorX_3=\{c\}\) we gain \(P^{\modelM}_{\backdoorX_3
    \subseteq \backdoorX}=\SB a \rsep g \SE\) and the least model
  \(L=\{a,g\}\) of \(\DH( P^{\modelM}_{\backdoorX_3 \subseteq
    \backdoorX})\). Since Condition~4c holds, the algorithm returns
  \AlgTrue for \(X_3\).
  For \(\backdoorX_4=\{b,c\}\) we obtain \(P^\modelM_{\backdoorX_4
    \subseteq \backdoorX}= \{ g \lrsep \} \). The set \(L=\{g\}\) is
  the least model of \(\DH(P^{\modelM}_{\backdoorX_4 \subseteq
    \backdoorX})\). Since Condition~4c holds, the algorithm returns
  \AlgTrue for \(X_4\).
%
%
  For all remaining subsets of $X$ the Algorithm \MinCheck returns
  \AlgTrue according to Condition~1.
  Consequently, \(\modelM\) is a minimal model of \(P^{\modelM}\) and
  thus an answer set of $P$.
\end{EX}

Next, we state and prove that there are fpt-reductions from
\BdBrave{\Normal[0]} and \BdSkeptHyph{\Normal[0]} to \SAT which is the
main result of this paper.

\begin{THE}\label{the:fpt-SAT}
  Given a disjunctive logic program $P$ of input size $n$, a {strong
    {\normalfont\Normal}-backdoor} $\backdoorX$ of $P$ of size $k$,
  and an atom $a^* \in \at(P)$, we can produce in time $O(2^k n^2)$
  propositional formulas $\FBrave(a^*)$ and $\FSkept(a^*)$ such that
  (i)~$\FBrave(a^*)$ is satisfiable if and only if $a^*$ is in some
  answer set of $P$, and (ii)~$\FSkept(a^*)$ is unsatisfiable if and
  only if $a^*$ is in all answer sets of $P$.
\end{THE}

\begin{proof}
  We would like to use a similar approach as in the proof of
  Theorem~\ref{the:fpt-check}. However, we cannot consider all
  possible models $\modelM$ one by one, as there could be too many of
  them.  Instead, we will show that it is possible to implement
  $\MinCheck(X_1)$ for each set $X_1 \subseteq X$ nondeterministically
  in such a way that we do not need to know $\modelM$ in
  advance. Possible sets $\modelM$ will be represented by the truth
  values of certain variables, and since the truth values do not need
  to be known in advance, this will allow us to consider all possible
  sets $\modelM$ without enumerating them.

  Next, we describe the construction of the formulas~$\FBrave(a^*)$
  and $\FSkept(a^*)$ in detail.  

%
%
  Among the variables of our formulas will be a set $V:=\SB v[a] \SM a
  \in \at(P)\SE$ containing a variable for each atom of~$P$. The truth
  values of the variables in $V$ represent a subset $\modelM\subseteq
  \at(P)$, such that $v[a]$ is true if and only if $a\in \modelM$.
 
We define 
\medskip

\shortversion{\noindent $\FBrave(a^*) := F^\text{mod} \wedge F^\text{min} \wedge
v[a^*]$ and \medskip

\noindent $\FSkept(a^*) := F^\text{mod} \wedge F^\text{min} \wedge
\neg v[a^*]$, }
\longversion{
\begin{align*}\FBrave(a^*) &:= F^\text{mod} \wedge F^\text{min} \wedge
v[a^*] \quad\text{and}
\end{align*}
\begin{align*}
\FSkept(a^*) &:= F^\text{mod} \wedge F^\text{min} \wedge
\neg v[a^*],
\end{align*}
}
\medskip

\noindent where $F^\text{mod}$ and $F^\text{min}$ are formulas,
defined below, that check whether the truth values of the variables in
$V$ represent a model $\modelM$ of $P^\modelM$, and whether $\modelM$
is a minimal model of $P^\modelM$, respectively.

The definition of $F^\text{mod}$ is easy:

\shortversion{\noindent $F^{\text{mod}}:= \bigwedge_{r\in P} \Big(\bigwedge_{b\in
  B^-(r)} \neg v[b] \rightarrow$

\hfill $\big( \bigvee_{b\in B^+(r)} \neg v[b] \vee \bigvee_{b\in H(r)}
v[b]\big)\Big)$.}
\longversion{
  \begin{align*}
    F^{\text{mod}}:= \bigwedge_{r\in P} \Big(\bigwedge_{b\in
      B^-(r)} \neg v[b] \rightarrow \big( \bigvee_{b\in B^+(r)} \neg v[b]
    \vee \bigvee_{b\in H(r)} v[b]\big)\Big).
\end{align*}
}

The definition of $F^\text{min}$ is more involved. First we define:
\medskip

\shortversion{\noindent $F^{\text{min}}:= \bigwedge_{1\leq i \leq 2^k}
F^\text{min}_i$, }
\longversion{
  \begin{align*}
    F^{\text{min}}:= \bigwedge_{1\leq i \leq 2^k} F^\text{min}_i,
  \end{align*}
}
\medskip

\noindent where $F_i^\text{min}$, defined below, encodes the Algorithm
$\MinCheck(X_i)$ for each set $\backdoorX_i$ where $\backdoorX_1,
\dots, \backdoorX_{2^k}$ is an enumeration of all the subsets
of~$\backdoorX$.
  
The formula~$F_i^\text{min}$ will contain, in addition to the
variables in $V$, $p$ distinct variables for each atom of~$P$, $p :=
\min\{\Card{P}, \Card{\at(P)}\}$.  In particular, the set of variables
of $F_i^\text{min}$ is the disjoint union of $V$ and $U_i$ where
$U_i:=\SB u_i^j[a] \SM a \in \at(P)$, $1\leq j \leq p \SE$. We write
$U_i^j$ for the subset of $U_i$ containing all the variables
$u_i^j[a]$. We assume that for $i\neq i'$ the sets $U_i$ and $U_{i'}$
are disjoint.  For each $a\in \at(P)$ we also use the propositional
constants $\backdoorX(a)$ and $\backdoorX_1(a)$ that are true if and
only if $a\in \backdoorX$ and $a\in \backdoorX_1$, respectively.

\longversion{The truth values of the variables in $U_i^p$ represent
  the unique minimal model of $\DH(P^\modelM_{\backdoorX_s\subseteq
    \backdoorX})$.}

  We define the formula~$F_i^\text{min}$ by means of the following
  auxiliary formulas.

  The first auxiliary formula checks whether the truth values of the
  variables in~$V$ represent a set~$\modelM$ that contains $X_i$:

  \medskip 
\shortversion{
  \noindent $ F^{\subseteq}_i:= \bigwedge_{a\in \backdoorX}
  \backdoorX_i(a) \rightarrow v[a]$.
}
\longversion{
  \begin{align*}
    F^{\subseteq}_i:= \bigwedge_{a\in \backdoorX}
    \backdoorX_i(a) \rightarrow v[a].
  \end{align*}
}

\medskip
The next auxiliary formula encodes the computation of the least model
(``lm'')~$L$ of $\DH(P^\modelM_{\backdoorX_i\subseteq \backdoorX})$
where $\modelM$ and $L$ are represented by the truth values of the
variables in $V$ and $U_i^p$, respectively.

\medskip
\shortversion{\noindent $F^\text{lm}_{i} := \bigwedge_{a\in \at(P), 0\leq i \leq p}
F_{i}^{(a,i)}$, where}
\longversion{
  \begin{align*}
    F^\text{lm}_{i} := \bigwedge_{a\in \at(P), 0\leq i \leq p}
    F_{i}^{(a,i)}, \quad\text{where}
  \end{align*}
}

\shortversion{\noindent $F_{i}^{(a,0)} := u^{0}_i[a] \leftrightarrow \false$,}
\longversion{
  \begin{align*}
    F_{i}^{(a,0)} := u^{0}_i[a] \leftrightarrow \false,
  \end{align*}
}

\shortversion{\noindent $F_{i}^{(a,j)} := u_i^{j}[a] \leftrightarrow
\big[u_i^{j-1}[a] \vee\bigvee_{r\in P_{\backdoorX_i\subseteq
    \backdoorX}, a\in H(r)} $

\hfill$(\bigwedge_{b\in B^+(r)} u_{i}^{j-1}[b] \wedge \bigwedge_{b\in
  B^-(r)} \neg v[b])\big]$

\hfill (for $1\leq j \leq p-1$).}
\longversion{
  \begin{align*}
    F_{i}^{(a,j)} := u_i^{j}[a] \leftrightarrow \big[u_i^{j-1}[a] \vee
    \bigvee_{r\in P_{\backdoorX_i\subseteq \backdoorX}, a\in H(r)} 
    (\bigwedge_{b\in B^+(r)} u_{i}^{j-1}[b] \wedge 
    \bigwedge_{b\in B^-(r)} \neg v[b])\big]\\
    \hfill\text{(for } 1\leq j \leq p-1\text{)}.
  \end{align*}
}
\medskip

The idea behind the construction of $F^\text{lm}_{i}$ is to simulate
the linear-time algorithm of \citex{DowlingGallier84}. Initially, all
variables are set to false. This is represented by
variables~$u_i^0[a]$. Now we flip a variable from false to true if and
only if there is a Horn rule where all the variables in the rule body
are true. We iterate this process until a fixed-point is reached, then
we have the least model. The flipping is represented in our formula by
setting a variable~$u_i^{j}[a]$ to true if and only if either
$u_i^{j-1}[a]$ is true, or there is a rule~$r\in \DH( P^\modelM_{
  \backdoorX_i \subseteq \backdoorX})$ such that $H(r)=\{a\}$ and
$u_i^j[b]$ is true for all $b\in B^+(r)$. The truth values of the
variables~$u_i^p$ now represent the least model of
$\DH(P^\modelM_{\backdoorX_i\subseteq \backdoorX})$.

The next four auxiliary formulas check whether the respective
condition (a)--(d) of algorithm\linebreak $\MinCheck(X_i)$ does not hold
for~$L$.

$F^{\text{(a)}}_{i}$ expresses that there is a rule in
$\Constraints(P^\modelM_{\backdoorX_i\subseteq \backdoorX})$ that is
not satisfied by~$L$:

\medskip
\shortversion{\noindent $F^{\text{(a)}}_{i} := \bigvee_{r\in
  P_{\backdoorX_i\subseteq \backdoorX}, H(r)\subseteq \backdoorX }
(\bigwedge_{b\in B^-(r)} \neg v[b] \wedge$

\hfill$\bigwedge_{b\in B^+(r)}  u_i^p[b] )$.}
\longversion{
  \begin{align*}
    F^{\text{(a)}}_{i} := \bigvee_{r\in
      P_{\backdoorX_i\subseteq \backdoorX}, H(r)\subseteq \backdoorX }
    (\bigwedge_{b\in B^-(r)} \neg v[b] \wedge \bigwedge_{b\in B^+(r)}
    u_i^p[b] ).
  \end{align*}
}
\medskip

$F^{\text{(b)}}_{i}$ expresses that $L$ contains an atom that is not
in $\modelM\setminus \backdoorX$:

\medskip
\shortversion{\noindent $F^{\text{(b)}}_{i} := \bigvee_{a\in \at(P)\setminus
  \backdoorX} (\neg v[a] \wedge u_i^p[a])$.}
\longversion{
  \begin{align*}
    F^{\text{(b)}}_{i} := \bigvee_{a\in \at(P)\setminus \backdoorX} 
    (\neg v[a] \wedge u_i^p[a]).
  \end{align*}
}
\medskip

$F^{\text{(c)}}_{i}$ expresses that $L \cup \backdoorX_i$ equals
$\modelM$ or $L \cup \backdoorX_i$ contains an atom that is not
in~$\modelM$: \medskip

\shortversion{\noindent $F^{\text{(c)}}_{i} := \left(\bigwedge_{a\in \at(P)} v[a]
  \leftrightarrow (u_i^p[a] \vee \backdoorX_i(a)) \right) \vee$

\hfill $ \left(\bigvee_{a\in \at(P)} ( u_i^p[a] \vee \backdoorX_i(a))
  \wedge \neg v[a] \right)$.}
\longversion{
  \begin{align*}
    F^{\text{(c)}}_{i} := \left(\bigwedge_{a\in \at(P)} v[a]
      \leftrightarrow (u_i^p[a] \vee \backdoorX_i(a)) \right) \vee
    \left(\bigvee_{a\in \at(P)} ( u_i^p[a] \vee \backdoorX_i(a))
      \wedge \neg v[a] \right).
  \end{align*}
}
\medskip

$F^{\text{(d)}}_{i}$ expresses that $P^\modelM$ contains a
rule that is not satisfied by $L \cup \backdoorX_i$: \medskip

\shortversion{\noindent $F^{\text{(d)}}_{i} := \bigvee_{r\in P} [ \bigwedge_{a\in
  B^-(r)} \neg v[a] \wedge$

\hfill $\bigwedge_{a\in H(r)} ( \neg u_i^p[a] \wedge \neg
\backdoorX_i(a)) \wedge \bigwedge_{b\in B^+(r)} ( u_i^p[b] \vee
\backdoorX_i(b))].$}
\longversion{
  \begin{align*}
    F^{\text{(d)}}_{i} := \bigvee_{r\in P} [ \bigwedge_{a\in B^-(r)}
    \neg v[a] \wedge \bigwedge_{a\in H(r)} ( \neg u_i^p[a] \wedge \neg
    \backdoorX_i(a)) \wedge \bigwedge_{b\in B^+(r)} ( u_i^p[b] \vee
    \backdoorX_i(b))].
  \end{align*}
}
\medskip

\noindent Now we can put the auxiliary formulas together and obtain
\medskip
 
\shortversion{$F^{\text{min}}_{i}:=
\neg F^{\subseteq}_i 
\vee (F^{\text{lm}}_i
\wedge
(F^{\text{(a)}}_{i} \vee
F^{\text{(b)}}_{i} \vee
F^{\text{(c)}}_{i} \vee
F^{\text{(d)}}_{i}))$.}
\longversion{
 \begin{align*}
   F^{\text{min}}_{i}:= \neg F^{\subseteq}_i  \vee (F^{\text{lm}}_i \wedge
   (F^{\text{(a)}}_{i} \vee F^{\text{(b)}}_{i} \vee F^{\text{(c)}}_{i}
   \vee F^{\text{(d)}}_{i})).
 \end{align*}
}
\medskip

It follows by Lemma~\ref{lem:soundness} and by the construction of the
auxiliary formulas that (i)~$\FBrave(a^*)$ is satisfiable if and only
if $a^*$ is in some answer set of~$P$, and (ii)~$\FSkept(a^*)$ is
unsatisfiable if and only if $a^*$ is in all answer sets of $P$.

Hence, it remains to observe that for each $i\leq 2^k$ the auxiliary
formula $F^{\text{lm}}_{i}$ can be constructed in quadratic time,
whereas the auxiliary formulas $F^{\subseteq}_{i}$ and
$F^{\text{\tiny{(a)}}}_{i} \vee F^{\text{\tiny{(b)}}}_{i} \vee
F^{\text{\tiny{(c)}}}_{i} \vee F^{\text{\tiny{(d)}}}_{i}$ can be
constructed in linear time. Since $\Card{ \backdoorX } = k$ by
assumption, we need to construct $O(2^k)$ auxiliary formulas in order
to obtain $\FSkept(a^*)$ and $\FBrave(a^*)$. Hence, the running time
as claimed in Theorem~\ref{the:fpt-SAT} follows\longversion{ and the
  theorem is established}.
\end{proof}

We would like to note that Theorem~\ref{the:fpt-SAT} remains true if
we require that the formulas $\FSkept(a^*)$ and $\FBrave(a^*)$ are in
Conjunctive Normal Form (CNF), as we can transform in linear time any
propositional formula into a satisfiability-equivalent formula in CNF,
e.g., using the well-known transformation due to
Tseitin~\shortversion{\shortcite{Tseitin68}.}\longversion{\shortcite{Tseitin68transl},
  see also \cite{KleineBuningLettman99}. This transformation produces
  for a given propositional formula $F'$ in linear time a CNF formula
  $F$ such that both formulas are equivalent with respect to their
  satisfiability, and the length of $F$ is linear in the length
  of~$F'$.}

Furthermore, the SAT encoding can be improved. For instance, one could
share parts between the formulas $F_i^{\text{min}}$ or replace the
quadratic formula $F^\text{lm}_{i}$ for the computation of least
models with a smaller and more sophisticated \SAT
encoding~\cite{Janhunen04} or a \problemn{Sat(Dl)}
encoding~\cite{JanhunenNiemelaSevalnev09} for the \textsc{Smt}
framework which combines propositional logic and linear constraints.

We would like to point out that our approach directly extends to more
general problems, when we look for answer sets that satisfy a certain
global property which can be expressed by a propositional formula
$F^{\text{prop}}$ on the variables in $V$. We just check the
satisfiability of $F^{\text{mod}} \wedge F^{\text{min}} \wedge
F^{\text{prop}}$.

\begin{EX} 
  \sloppypar Consider the program $P$ from Example~\ref{ex:running}
  and the \strongBds{\Normal} $X=\{b,c,h\}$ of $P$ from
  Example~\ref{ex:backdoor}. We ask whether the atom $b$ is contained
  in at least one answer set. To decide the question, we check that
  $F_\text{brave}(b)$ is satisfiable and we answer the question
  positively. Since $\modelM=\{b,c,g\}$ is model of $P^\modelM$ we can
  satisfy $F^\text{mod}$ with a truth assignment $\tau$ that maps~$1$
  to each variable $v[x]$ where $x \in \{b, c, g\}$ and $0$ to each
  variable $v[x]$ where $x \in \at(P) \setminus \{b, c, g\}$. For
  $i=1$ let $X_1 = \emptyset$. Then we have for the constants
  $X_1(x)=0$ where $x \in \{b, c, h\}$. Observe that $\tau$ already
  satisfies $F^{\subseteq}_i$ and that $F^\text{lm}_i$ encodes the
  computation of the least model $L$ of $\DH( P^\modelM_{\backdoorX_1
    \subseteq \backdoorX})$ where $L$ is represented by the truth
  values of the variables in $U^P_i=\SB u^p_i[x] \SM x \in \at(P)
  \SE$. Thus $\tau$ also satisfies $F^\text{lm}_i$ if $\tau$ maps
  $u^p_i[a]$ to $1$, $u^p_i[g]$ to $1$, and $u^p_i[x]$ to $0$ where $x
  \in \at(P) \setminus \{a, g\}$. As $\tau$ satisfies
  $F^\text{\tiny{(a)}}_1$, the truth assignment~$\tau$ satisfies the
  formula~$F^\text{min}_1$. It is not hard to see that
  $F^\text{min}_i$ is satisfiable for other values of $i$. Hence the
  formula~$F_\text{brave}(b)$ is satisfiable and $b$ is contained in
  at least one answer set.
\end{EX}

\subsection*{Completeness for $\paraNP$ and $\coparaNP$}
The parameterized complexity class $\paraNP$ contains all
parameterized decision problems $L$ such that $(I,k)\in L$ can be
decided \emph{nondeterministically} in time $O(f(k)\CCard{I}^c)$, for
some computable function~$f$ and constant~$c$~\cite{FlumGrohe06}. By
$\coparaNP$ we denote the class of all parameterized decision problems
whose complement (the same problem with yes and no answers swapped) is
in $\paraNP$.

\shortversion{As a corollary to Theorem~\ref{the:fpt-SAT} we
    obtain the following result:}

  \longversion{If a non-parameterized problem is NP-complete, then
    adding a parameter that makes it $\paraNP$\hy complete does not
    provide any gain\shortversion{.}\longversion{, as this holds even
      true if the parameter is the constant~0.} Therefore a
    $\paraNP$\hy completeness result for a problem that without
    parameterization is in $\NP$, is usually considered as an utterly
    negative result. However, if the considered problem without
    parameter is outside $\NP$, and we can show that with a suitable
    parameter the problem becomes $\paraNP$\hy complete, this is in
    fact a positive result.
  Indeed, we get such a positive result as a corollary to
  Theorem~\ref{the:fpt-SAT}.}

\begin{COR}\label{cor:para}
  \BdBrave{\Normal[0]} is $\paraNP$-com\-plete, and
  \BdSkept{\Normal[0]} is $\coparaNP$-com\-plete.
\end{COR}
\begin{proof}If a parameterized problem $L$ is $\NP$\hy hard when we
  fix the parameter to a constant, then $L$ is $\paraNP$-hard
  (\citey{FlumGrohe06}, Th.~2.14). As \BdBrave{\Normal[0]} is $\NP$\hy
  hard for backdoor size~$0$, we conclude that \BdBrave{\Normal[0]} is
  $\paraNP$-hard. A similar argument shows that \BdSkept{\Normal[0]}
  is $\coparaNP$-hard. \SAT, considered as a pa\-ram\-e\-ter\-ized
  problem with constant parameter~$0$, is clearly $\paraNP$\hy
  complete, this also follows from the mentioned result of Flum and
  Grohe~\shortcite{FlumGrohe06}; hence \textsc{UnSat} is
  $\coparaNP$\hy complete. As Theorem~\ref{the:fpt-SAT}
  \mbox{provides} fpt-re\-duc\-tions from \BdBrave{\Normal[0]} to
  \SAT, and from \BdSkept{\Normal[0]} to \textsc{UnSat}, we conclude
  that \BdBrave{\Normal[0]} is in $\paraNP$, and \BdSkept{\Normal[0]}
  is in $\coparaNP$.
\end{proof}

\section{Finding Backdoors}  
In this section, we study the problem of finding backdoors, formalized
in terms of the following parameterized problem:
\shortversion{\BdDetect{Strong}{\ensuremath{\CCC}}: \emph{Given} a
  (disjunctive) program~$P$, and the parameter integer
  $k$. \emph{Question:} Find a strong $\CCC$\hy backdoor~$X$ of $P$ of
  size at most~$k$, or report that such $X$ does not
  exist.}\longversion{\pproblem{\BdDetect{Strong}{\ensuremath{\CCC}}}{A
    (disjunctive) program~$P$, and an integer $k$.}{The
    integer~$k$.}{Find a strong $\CCC$\hy backdoor~$X$ of $P$ of size
    at most~$k$, or report that such $X$ does not exist.}}
We also consider the problem \BdDetect{Deletion}{\ensuremath{\CCC}},
defined similarly.

Let $P$ be a program. Let the \emph{head dependency graph} $U^H_P$ be
the undirected graph~$U^H_P=(V,E)$ defined on the set~$V=\at(P)$ of
atoms of the given program~$P$, where two atoms~$x,y$ are joined by an
edge~$xy \in E$ if and only if $P$ contains a non-tautological
rule~$r$ with $x,y\in H(r)$.  A \emph{vertex cover} of a graph
$G=(V,E)$ is a set $X\subseteq V$ such that for every edge~$uv\in E$
we have $\{u,v\}\cap \backdoorX \neq \emptyset$.

\begin{LEM}\label{lem:normal-vc}
Let $P$ be a program. A set $\backdoorX \subseteq \at(P)$ is a
\delBds{\ensuremath{\Normal}} of $P$ if and only if $\backdoorX$ is a vertex
cover of $U^H_P$.
\end{LEM}
\shortversion{Due to space limitations we omit the proof.}
\longversion{
\begin{proof}
Let $\backdoorX$ be a \delBds{\Normal} of $P$. Consider an edge $uv$ of
$U^H_P$, then there is a rule $r\in P$ with $u,v\in H(r)$ and $u\neq v$.
Since $\backdoorX$ is a deletion $\Normal$-backdoor set of $P$, we have
$\{u,v\}\cap \backdoorX \neq \emptyset$. We conclude that $\backdoorX$ is a
vertex cover of $U^H_P$.

Conversely, assume that $\backdoorX$ is a vertex cover of $U^H_P$. Consider
a rule $r\in P-\backdoorX$ for proof by contradiction. If $\Card{H(r)}\geq
2$ then there are two variables $u,v\in H(r)$ and an edge $uv$ of $U^H_P$
such that $\{u,v\}\cap \backdoorX =\emptyset$, contradicting the assumption
that $\backdoorX$ is a vertex cover. Hence the lemma prevails.
\end{proof}}

\begin{THE}\label{the:horn}	
The problems \BdDetect{Strong}{\Normal} and \BdDetect{Deletion}{\Normal} are
fixed-parameter tractable.
\longversion{In particular, given a program $P$ of input size $n$, and an integer
$k$, we can find in time $O(1.2738^k + kn)$ a \strongBds{\Normal} 
of $P$ with a size $\leq k$ or decide that no such backdoor exists.}
\end{THE}
\begin{proof}
  In order to find a \delBds{\Normal} of a given program~$P$, we use
  Lemma~\ref{lem:normal-vc} and find a vertex cover of size at
  most~$k$ in the head dependency graph~$U^D_P$. A vertex cover of
  size~$k$, if it exists, can be found in time~\mbox{$O(1.2738^k +
    kn)$}~\cite{ChenKanjXia06}. Thus the theorem holds for
  \delBds{\Normal}s. Lemma~\ref{lem:normal-strong-deletion} states
  that the \strongBds{\Normal}s of $P$ are exactly the
  \delBds{\Normal}s of $P$ (as we assume that $P$ does not contain any
  tautological rules). The theorem follows.
\end{proof}

\sloppypar In Theorem~\ref{the:fpt-SAT} we assume that a
\strongBds{\Normal} of size at most~$k$ is given when solving the
prob\-lems \problemn{Strong}
\Normal-\problemn{Back\-door-Brave-Reasoning} and
\problemn{Skeptical-Reasoning}. As a direct consequence of
Theorem~\ref{the:horn}, this assumption can be dropped, and we obtain
the following corollary.

\begin{COR}
  The results of Theorem~\ref{the:fpt-SAT} and
  Corollary~\ref{cor:para} still hold if the backdoor is not given as
  part of the input.
\end{COR} 
                         
\section{Backdoors to Tightness} 

We associate with each program \(P\) its \emph{positive dependency
  graph} \(D^+_P\). It has the atoms of~$P$ as vertices and a directed
edge \((x,y)\) between any two atoms \(x,y \in \at(P)\) for which
there is a rule \(r\in P\) with \(x\in H(r)\) and \(y\in B^+(r)\). A
program is called \emph{tight} if \(D^+_P\) is
acyclic~\cite{LeeLifschitz03}.  We denote the class of all tight
programs by \tight.

It is well known that the main $\ASP$ reasoning problems are in $\NP$
and $\coNP$ for tight programs; in fact, a reduction to \SAT based on
the concept of \emph{loop formulas} has been proposed by Lin and
Zhao~\shortcite{LinZhao04a}.  This was then generalized by Lee and
Lifschitz~\shortcite{LeeLifschitz03} with a reduction that takes as
input a disjunctive normal program~$P$ together with the set $S$ of
all directed cycles in the positive dependency graph of~$P$, and
produces a CNF formula~$F$ such that answer sets of $P$ correspond to
the satisfying assignments of $F$. This provides an fpt-reduction from
the problems \problemn{Brave Reasoning} and \problemn{Skeptical
  Reasoning} to \SAT, when parameterized by the number of all cycles
in the positive dependency graph of a given program $P$, assuming that
these cycles are given as part of the input.

The number of cycles does not seem to be a very practical parameter,
as this number can quickly become very large even for very simple
programs.  Lifschitz and Razborov~\shortcite{LifschitzRazborov06} have
shown that already for normal programs an exponential blowup may
occur, since the number of cycles in a normal program can be
arbitrarily large. Hence, it would be interesting to generalize the
result of Lee and Lifschitz~\shortcite{LeeLifschitz03} to a more
powerful parameter. In fact, the size $k$ of a \delBds{\tight} would
be a candidate for such a parameter, as it is easy to see, it is at
most as large as the number of cycles, but can be exponentially
smaller. This is a direct consequence of the following two
observations: (i)~If a program~$P$ has exactly~$k$ cycles in $D^+_P$,
we can construct a \delBds{\tight}~$\backdoorX$ of $P$ by taking one
element from each cycle into $\backdoorX$. (ii)~If a program~$P$ has a
\delBds{\tight} of size~$1$, it can have arbitrarily many cycles that
run through the atom in the backdoor.

\longversion{In the following,}\shortversion{Next,} we show that this
parameter~$k$ is of little use, as the reasoning problems already
reach their full complexity for programs with a \delBds{\tight} of
size~$1$.

\begin{THE}\label{the:tighthardness} 
  The problems \BdBrave{\tight} and \BdSkept{\tight} are
  $\Sigma^P_2$-hard and $\Pi^P_2$-hard, respectively, even for programs
  that admit a \strongBds{\tight} of size~$1$, and the backdoor is
  provided with the input. The problems remain hard when we consider a
  \delBds{\tight} instead of a \strongBds{\tight}.
\end{THE}
\shortversion{
  \begin{proof}
    \citex{EiterGottlob95} give a reduction from a
    $\Sigma_2^P$-complete ($\Pi^2_P$-complete) problem to Brave
    (Skeptical) Reasoning, where the produced program has rules of the
    form $x_i\por v_i \rsep y_i \por z_j \rsep y_j \leftarrow w \rsep
    z_j \leftarrow w\rsep w \leftarrow y_j,z_j\rsep w \leftarrow
    g(l_{k,1}), g(l_{k,2}),g(l_{k,3})\rsep w \leftarrow \pnot w
    \lrsep$. The set~$\{w\}$ is a \delBds{\tight} (\strongBds{\tight})
    of size 1.
\end{proof}
}
\longversion{
\begin{proof}
  Consider the reduction from Eiter and Gottlob~\cite{EiterGottlob95}
  which reduces the $\Sigma^P_2$-hard problem
  \problemn{$\exists\forall$-QBF Model Checking} to the problem
  \problemn{Consistency} (which decides whether given a program~$P$
  has an answer set). A $\exists\forall$ quantified boolean formula
  (QBF) has the form $\exists x_1 \dotsm \exists x_n\forall y_1 \dotsm
  \forall y_m D_1 \vee \dotsc \vee D_r$ where each $D_i=l_{i,1}\wedge
  l_{i,2}\wedge l_{i,3}$ and $l_{i,j}$ is either an atom
  $x_1,\dotsc,x_n,y_1,\dotsc,y_m$ or its negation. Their construction
  yields a program $P:=\{x_i\por v_i \rsep y_i \por z_j \rsep y_j
  \leftarrow w \rsep z_j \leftarrow w\rsep w \leftarrow y_j,z_j\rsep w
  \leftarrow g(l_{k,1}), g(l_{k,2}),g(l_{k,3})\rsep w \leftarrow \pnot
  w \lrsep\}$ for each $i \in \{1,\dotsc,n\}$, $j \in \{1,\dotsc,m\}$,
  $k \in \{1, \dotsc,r\}$, and $g$ maps as follows $g(\neg x_i)=v_i$,
  $g(\neg y_j)=z_j$, and otherwise $g(l)=l$. Since $P_{w=0}=\{x_i\por
  v_i\leftarrow \rsep y_j \vee z_j \lrsep \}$ and $P_{w=1}=\{x_i \por
  v_i \rsep y_j \vee z_j \rsep y_j \rsep z_j \rsep\}$ are both in
  $\tight$, the set $X=\{w\}$ is a \strongBds{\tight} of $P$ of size
  $1$. Thus the restriction does not yield tractability. The
  intractability of \problemn{Skeptical Reasoning} follows directly by
  the reduction of Eiter and Gottlob~\cite{EiterGottlob95} from the
  problem \problemn{Consistency}. Hardness of the other problems can
  be observed easily. Since $P - \{ w \}:=\{x_i\por v_i \rsep y_i \por
  z_j \rsep y_j \rsep z_j \rsep \leftarrow y_j,z_j\rsep \leftarrow
  g(l_{k,1}), g(l_{k,2}),g(l_{k,3})\rsep\}$ for each $i \in
  \{1,\dotsc,n\}$, $j \in \{1,\dotsc,m\}$, $k \in \{1, \dotsc,r\}$ is
  tight, we obtain a \delBds{\tight} of size $1$. In consequence we
  established the theorem.
\end{proof}
}

\section{Experiments}
Although our main results are theoretical, we have performed first
experiments to determine the size of smallest \strongBds{\Normal}s for
answer set programs representing \emph{structured} and \emph{random}
sets of instances. Our experimental results summarized in
Table~\ref{tab:normal} indicate, as expected, that structured
instances have smaller backdoors than random instances. As instances
from \texttt{ConformantPlanning} have rather small backdoors our
translation seems to be feasible for these instances. Furthermore, we
have compared the size of a smallest \strongBds{\Normal} with the size
of a smallest \strongBds{\textbf{Horn}}~\cite{FichteSzeider12a} for
selected sets. It turns out that for \texttt{ConformantPlanning}
smallest \strongBds{\Normal}s are significantly smaller (0.7\%
vs. 8.8\% of the total number of atoms).

\begin{table}
\centering
\longversion{\begin{tabular*}{0.75\textwidth}{@{\extracolsep{\fill} } lrrr}}
  \shortversion{\begin{tabular}{@{}lrrr@{}}}
    \toprule
    instance set & atoms & \shortversion{bd}\longversion{backdoor} (\%)
    & stdev \\
    \midrule
    \texttt{ConformantPlanning} & 1378.21 & 0.69 & 0.39\\
  \texttt{MinimalDiagnosis} & 97302.5 & 14.19 & 3.19\\
  \texttt{MUS} & 49402.3 & 1.90 & 0.35\\
  \texttt{StrategicCompanies} & 2002.0 & 6.03 & 0.04 \\
  \texttt{Mutex} & 6449.0 & 49.94 & 0.09 \\
  \texttt{RandomQBF} & 160.1 &49.69 & 0.00 \\
\bottomrule
\longversion{\end{tabular*}}
\shortversion{\end{tabular}}
\caption{
  Size of smallest \strongBds{\Normal}\shortversion{ (bd) } for benchmark
  sets, given as \% of the total number of atoms by the 
  mean over the instances. \shortversion{\\}
  \small{  
    \texttt{ConformantPlanning}: secure planning under incomplete
    initial states~\longversion{\protect}\cite{ToPontelliSon09} encodings provided by
    Gebser and Kaminski~\longversion{\protect}\shortcite{GebserKaminski12}.
%
%
    \texttt{MinimalDiagnosis}: an application in systems
    biology~\longversion{\protect}\cite{GebserSchaubThieleUsadelVeber08}
    instances provided by~\longversion{\protect}\citex{CalimeriEtAl11}.
    \texttt{MUS}: problem whether a clause belongs
    to some minimal unsatisfiable 
    subset~\longversion{\protect}\cite{JanotaMarques-Silva11} 
    encoding provided by Gebser and 
    Kaminski~\longversion{\protect}\shortcite{GebserKaminski12}.   
    \texttt{StrategicCompanies}: encoding the $\Sigma^P_2$-complete 
    problem of producing and owning companies and strategic sets 
    between the companies~\longversion{\protect}\cite{GebserLiuNamasivayamNeumannSchaubTruszczynski07}.
    \texttt{Mutex}: equivalence test of partial implementations of
    circuits, provided
    by~\longversion{\protect}\citex{MarateaRiccaFaberLeone08} 
    based on QBF instances of~\longversion{\protect}\citex{AyariBasin00}.
    \texttt{RandomQBF}: translations of randomly generated $2$-QBF
    instances using the method by Chen and
    Interian~\longversion{\protect}\shortcite{ChenInterian05}
    instances provided by
    Gebser~\longversion{\protect}\shortcite{GebserLiuNamasivayamNeumannSchaubTruszczynski07}.}
\shortversion{\belowcaptionskip=-1em}
}\label{tab:normal}
\end{table}

\section{Conclusion}

We have shown that backdoors of small size capture structural
properties of disjunctive \ASP instances that yield to a reduction of problem
complexity. In particular, small backdoors to normality admit an
fpt-translation from \ASP to \SAT and thus reduce the complexity of the
fundamental \ASP problems from the second level of the Polynomial
Hierarchy to the first level. Thus, the size of a smallest
$\Normal$\hy backdoor is a structural parameter that admits a
fixed-parameter tractable complexity reduction without making the
problem itself fixed-parameter tractable.

Our \emph{complexity barrier breaking reductions} provide a new way of
using fixed-parameter tractability and enlarges its applicability.  In
fact, our approach as exemplified above for \ASP is very general and
might be applicable to a wide range of other hard combinatorial
problems that lie beyond $\NP$ or $\coNP$. We hope that our work
stimulates further investigations into this direction such as the
application to abduction very recently established by
\citex{PfandlerRummeleSzeider13}.

Our first empirical results suggest that with an improved SAT encoding
and preprocessing techniques to reduce the size of $\Normal$\hy
backdoors (for instance,
\emph{shifting},~\citey{JanhunenOikarinenTompitsWoltran07}), our
approach could be of practical use, at least for certain classes of
instances, and hence might fit into a portfolio-based solver.

\shortversion{
\appendix
\begingroup
    \fontsize{9pt}{10pt}\selectfont
}

\newcommand{\publisher}[1]{\newblock #1}
\newcommand{\pagenums}[1]{#1}
\newcommand{\journalab}{Journal\xspace}

\shortversion{
\newcommand{\MKP}{Morgan Kaufmann}
\bibliographystyle{aaai}

\endgroup}

\longversion{
\bibliographystyle{named}

}
\end{document}